\documentclass[11pt]{article}

\usepackage[T1]{fontenc}
\usepackage{lmodern}
\usepackage[margin=1.05in]{geometry}
\usepackage{amsmath,amssymb,amsthm,mathtools,mathrsfs}
\usepackage{graphicx}
\usepackage{microtype}
\usepackage{xcolor}
\usepackage{float}
\usepackage[colorlinks=true,linkcolor=blue!55!black,citecolor=blue!55!black,urlcolor=blue!55!black]{hyperref}

\newtheorem{theorem}{Theorem}[section]
\newtheorem{lemma}[theorem]{Lemma}
\newtheorem{proposition}[theorem]{Proposition}

\theoremstyle{remark}
\newtheorem{remark}[theorem]{Remark}
\numberwithin{equation}{section}
\newtheorem{definition}{Definition}[section]

\newcommand{\C}{\mathbb C}
\newcommand{\R}{\mathbb R}
\newcommand{\Z}{\mathbb Z}
\newcommand{\calH}{\mathcal H}
\newcommand{\ind}{\operatorname{ind}}
\newcommand{\eps}{\varepsilon}

\newcommand{\ip}[2]{\langle #1,#2\rangle}

\graphicspath{{figures/}{outputs/pdf/figures/}}

\title{Exactly flat bands beyond the chiral limit in Bistritzer--MacDonald Hamiltonians}
\author{Zhen Huang\textsuperscript{1,2,3}\quad
Kevin D. Stubbs\textsuperscript{1,4}\\[0.5em]
\small \textsuperscript{1}Department of Mathematics, University of California,
Berkeley, CA 94720, USA\\
\small \textsuperscript{2} Center for Computational Mathematics, Flatiron Institute, New York, NY 10010, USA\\
\small \textsuperscript{3} Center for Computational Quantum Physics, Flatiron Institute, New York, NY 10010, USA\\
\small \textsuperscript{4} Department of Mathematics, University of Minnesota, Minneapolis, MN 55455, USA}
\date{\today}

\begin{document}
\maketitle

\begin{abstract}
The chiral limit of the Bistritzer--MacDonald model for twisted bilayer
graphene has exactly flat bands at magic angles.  We show that exact flatness
can persist beyond the chiral limit with nonchiral tunnelling potentials in the usual symmetry class.  We construct a family of nonchiral tunnelling
potentials for which the Hamiltonian has at least two zero-energy states at
every Bloch momentum and for every real value of the nonchiral coupling.  This gives a family of counterexamples to Open Problem~3 of
Zworski's survey \cite{ZworskiSurvey}. Conversely, we prove that every admissible potential with this property belongs to this family, thereby obtaining a complete characterization.
By contrast, we show that the standard BM model, namely, BM Hamiltonians with first harmonic
tunnelling potentials, does not admit exactly flat bands at the magic angle except possibly at a
discrete set of nonchiral coupling strengths, with no finite accumulation
point.  To prove this, we introduce an explicit invertibility criterion.  { The
criterion requires the projected perturbation to have a nonzero scalar coefficient
at some Bloch momentum, obstructing exact flatness for general nonchiral tunnelling.}  A local Taylor expansion argument verifies the
criterion for the standard BM potentials.  For the family constructed above,
however, the projected perturbation vanishes at every Bloch momentum, so the invertibility criterion fails.
\end{abstract}  
\section{Introduction}

Introduced in their seminal paper \cite{BistritzerMacDonald}, the Bistritzer--MacDonald (BM) model describes the low-energy electronic states twisted bilayer graphene (TBG), which is created by stacking two layers of graphene on top of each other a small relative twist between them.
At Bloch momentum \(k\), the BM model is described by a Hamiltonian which takes the form
\[
 H_k(\alpha,\lambda)=
 \begin{pmatrix}
  \lambda W_V & D(\alpha)^* + \overline{k} \\
  D(\alpha) + k&\lambda W_V
 \end{pmatrix},
\]
where \(D(\alpha)\)  and \(W_V\) are:
\begin{equation}\label{eq:D-and-W}
 D(\alpha)=
 \begin{pmatrix}
  2D_{\bar z}&\alpha U(z)\\
  \alpha U(-z)&2D_{\bar z}
 \end{pmatrix},
 \quad
 W_V(z)=
 \begin{pmatrix}
  0&V(z)\\ V(-z)&0 
 \end{pmatrix},
 \quad 
 U, V \in C^\infty(\mathbb{C})
\end{equation}
Here $D_{\bar z}=-\frac{1}{2}\mathrm{i} (\partial_{x} - \mathrm{i} \partial_y) $ is the Dirac operator, and \(z=x+\mathrm{i}y\) is a complex coordinate on the plane.  (see Section \ref{sec:model} for precise definitions).
The potentials $U, V \in C^\infty(\mathbb{R})$ model tunneling between the two layers of graphene; $U$ corresponds to same sublattice ($AA$ and $BB$) tunneling whereas $V$ corresponds to different sublattice ($AB$ and $BA$) tunneling.
If the two layers of graphene have relative twist $\theta$ then $\alpha \propto \theta^{-1}$, $\lambda$ measures the relative strength of the same sublattice tunneling and the different sublattice tunneling.

The key observation of Bistrizter--MacDonald is that when $\lambda = 1$, there exist choices of $\alpha$, referred to as ``magic'' angles, where the Fermi velocity of the low energy bands become nearly flat. 
Following up on this work, Tarnopolsky, Kruchkov, and Vishwanath \cite{TarnopolskyKruchkovVishwanath} showed that that when $\lambda = 0$, called the chiral limit, there exist values of $\alpha$ so that bands of the BM model are provably exactly flat.
This result was expanded upon by Watson and Luskin \cite{WatsonLuskin} who gave a rigorous computer-assisted proof that the first real magic angle exists as well as Becker, Embree, Wittsten and Zworski \cite{Becker2021SpectralCharacterizationMagic,Becker2022MathematicsMagicAngles} who characterized the set \(\mathcal A\) of magic parameters in terms of eigenvalues of a Birman--Schwinger operator.
The proof of \cite{Becker2021SpectralCharacterizationMagic,Becker2022MathematicsMagicAngles}, in particular showed that for \textit{any} choice of potential $U$ satisfying the appropriate symmetries then the construction of flat bands from \cite{TarnopolskyKruchkovVishwanath} easily generalizes.

Unfortunately, \textit{ab initio} calculations estimate that $\lambda \approx 0.5 - 0.9$ the chiral limit ($\lambda = 0$) is an idealization of the real-world physical system \cite{TarnopolskyKruchkovVishwanath,NamKoshino} and it has been widely observed numerically that the central bands acquire a nonzero width once \(\lambda \neq 0\) \cite{BistritzerMacDonald,NamKoshino,BeckerZworski}.
These numerical observatoin motivated the following \cite[Open Problem 3]{ZworskiSurvey}:
\begin{quote}
    Show that the [Bistritzer--Macdonald] Hamiltonian with $U$ and $V \not\equiv 0$ satisfying [symmetry constraints], cannot have flat bands when $\lambda \neq 0$. (Or give a counterexample to this claim.)
\end{quote}

Our first result provides a counterexample to this Open Problem.  Let \(U\) be a nonzero analytic
tunneling potential satisfying the symmetries of Section~\ref{sec:model}, let
\(\alpha_*\) be a simple real magic parameter for \(U\), and let \(p\) denote the first
component of a chiral zero mode (Lemma~\ref{lem:chiral-zero-mode}).  Set
\begin{equation}\label{eq:V0}
 V_0(z)=p(z)\,\overline{p(-z)},\quad V[f](z)=f(z)V_0(z),
\end{equation}
where \(f(z)\) is any real measurable function satisfying the
required symmetries and such that \(V[f]\) is essentially bounded.
This notation gives \(V[1]=V_0\), whereas \(V[0]\equiv0\).
Theorem~\ref{thm:exact} states that for every such \(V[f]\), every momentum \(k\) and
every \(\lambda \in \mathbb{R}\),
\begin{equation}\label{eq:exact}
 \dim_{\mathbb C}\ker H_k(\alpha_*,\lambda)\geq2 ;
\end{equation}
exact zero-energy flat bands therefore persist for arbitrarily large nonchiral coupling.
The mechanism is a topological one: the two chiral zero modes determine a pair of line
bundles which are preserved by the full Hamiltonian, and the restriction of the Hamiltonian to
their sum has Fredholm index two. 
Theorem~\ref{thm:characterization} shows that the
potentials \eqref{eq:V0} are the only ones with this property: every admissible \(V\)
for which \eqref{eq:exact} holds is of the form \(f V_0\) with
\(f\in\mathcal F_{\rm sym}\) real and measurable.

Our second result explains why this behavior is nonetheless exceptional, and confirms
the expectation behind Open Problem~3 for all other potentials.  At a simple magic
parameter \(\alpha_*\), let \(P_k\) denote the orthogonal projection onto the
two-dimensional space of chiral zero modes at momentum \(k\), and let
\(A_V(k):=P_k\mathbf W_VP_k\) be the compression of the nonchiral perturbation to that
space, a \(2\times2\) matrix.  We say that \(V\) satisfies the \emph{invertibility
criterion} if \(A_V(k_0)\) is invertible for some \(k_0\).
Proposition~\ref{prop:compression} shows that the zero modes at \(k_0\) then disappear
for every sufficiently small nonzero \(\lambda\); since zero-energy states persist at
\(k=\pm K\), the central bands cannot be flat.  More generally,
Theorem~\ref{thm:discrete-couplings} shows that for an admissible potential which is not
of the form \(fV_0\), exact flatness can occur only for a discrete set of real couplings
without accumulation points.  Finally, Proposition~\ref{prop:standard-BM} verifies the
invertibility criterion for the original Bistritzer--MacDonald potential
by a local Taylor expansion, so that this potential obstructs exact flatness for all but
a discrete set of nonchiral couplings.

Before proceeding with the proof, we review some of the previous works studying the existence of flat bands in TBG and related models.
Watson and Luskin \cite{WatsonLuskin} gave a rigorous computer assisted proof that the first magic angle found in \cite{TarnopolskyKruchkovVishwanath} exists.
In a pair of papers Becker, Embree, Wittsten, and Zworski \cite{Becker2021SpectralCharacterizationMagic,Becker2022MathematicsMagicAngles} showed that magic angles can be related to the eigenvalues of a Birman-Schwinger operator; these papers also showed for the anti-chiral limit $\alpha = 0$ and $\lambda \neq 0$ no flat bands exist. 
In later work, Becker, Humbert, and Zworski \cite{BHZfine, BHZintegrability,Becker2025DegenerateFlatBands} showed that the set of magic angles is infinite, proved that the number of flat bands is cannot be $0$ modulo 3, and calculated the Chern number of the associated bundle.
Closest to the present work, Becker and Zworski \cite{BeckerZworski} studied the eigenvalues of $\lambda \neq 0$ by applying perturbation theory chiral limit of the BM Hamiltonian.
Recent work by Becker, Quinn, Tao, Watson, and Yang \cite{Becker2026DiracConesMagic} has shown that there exists a measure zero set $\mathcal{S}$, described by analytic curves, so that the BM Hamiltonian with parameters $(\alpha, \lambda) \in \mathbb{R}^2 \setminus \mathcal{S}$ has a Dirac point and therefore cannot be flat.

In addition to the many works studying TBG, a number of works have studied proving the existence of flat bands in twisted graphene multilayers \cite{Li2025FlatBandsDirac,Becker2025ChiralLimitTwisted}  and for the scalar model of Galkowski and Zworski \cite{Galkowski2023AbstractFormulationFlat}, which is not equivalent to the chiral model but exhibits flat bands by a similar mechanism and is analyzed semiclassically in \cite{Dyatlov2026WKBStructureScalar}.

\subsection{Paper Organization}
This paper is organized as follows. Section~\ref{sec:model} gives the precise model and states the main results. Section~\ref{sec:construction} proves the exact flat-band construction, and Section~\ref{sec:nonflatness} proves the obstruction result and the application to the standard BM potentials.

\section*{Acknowledgments}
This work was supported in part by the Simons Targeted Grants in Mathematics and Physical Sciences on Moir\'e Materials Magic (Z.H., K.D.S.), and by the Applied Mathematics Program of the US Department of Energy (DOE) Office of Advanced Scientific Computing Research under contract number DE-AC02-05CH1123 (Z.H.).  The Flatiron Institute is a division
of the Simons Foundation. The authors thank Maciej Zworski and Lin Lin for valuable discussions and encouragement.

\section{Model and main results}\label{sec:model}

We use the standard complex-coordinate formulation of the model, identifying
a point \((x,y)\in\R^2\) with \(z=x+\mathrm{i}y\in\C\).  After
nondimensionalization, set
\[
 \omega=\mathrm{e}^{2\pi\mathrm{i}/3},\quad
 \Lambda=\Z\oplus\omega\Z,\quad X=\C/\Lambda.
\]
Here \(\Lambda\) is the moir\'e lattice, and \(X\) is the associated
real-space torus, also known as the moir\'e unit cell.
\textcolor{black}{We write \([z]:=z+\Lambda\in X\) for the equivalence class
of \(z\in\C\) modulo \(\Lambda\).}  Multiplication by \(\omega\) is
rotation through \(120^\circ\).  We use
\(K=4\pi/3\), identified with the vector \((4\pi/3,0)\).  For \(z,w\in\C\), write
\(\ip{z}{w}=\operatorname{Re}(z\bar w)\), and use
\[
  D_{\bar z}=-\mathrm{i}\partial_{\bar z},\quad
  D_z=-\mathrm{i}\partial_z,
  \quad
  \partial_{\bar z}=\tfrac12(\partial_x+\mathrm{i}\partial_y),\quad
  \partial_z=\tfrac12(\partial_x-\mathrm{i}\partial_y).
\]

For a chiral tunnelling potential \(U\) and a nonchiral tunnelling potential
\(V\), as in \eqref{eq:D-and-W}, define
\[
 D(\alpha)=
 \begin{pmatrix}
  2D_{\bar z}&\alpha U(z)\\
  \alpha U(-z)&2D_{\bar z}
 \end{pmatrix},
 \quad
 W_V(z)=
 \begin{pmatrix}
  0&V(z)\\ V(-z)&0
 \end{pmatrix}.
\]
Here \(D(\alpha)\) contains the chiral tunnelling channel, whereas \(W_V\) contains the nonchiral one.  
The full real-space Bistritzer--MacDonald Hamiltonian is
\begin{equation}\label{eq:H-BM}
 H_{\rm BM}(\alpha,\lambda)
 :=\begin{pmatrix}
  \lambda W_V&D(\alpha)^*\\
  D(\alpha)&\lambda W_V
 \end{pmatrix}:
 H^1(\C;\C^4)\longrightarrow L^2(\C;\C^4).
\end{equation}
Here \(D(\alpha)^*\) is the  adjoint of \(D(\alpha)\).  The off-diagonal blocks form the chiral Hamiltonian, while the diagonal term \(\lambda W_V\) is the nonchiral tunnelling term, which is artificially turned off in the chiral limit.   We will use the notation \(\mathbf W_V:=\operatorname{diag}(W_V,W_V)\) for the nonchiral part of the BM Hamiltonian.  

The tunnelling potentials $U$ and $V$ must respect symmetry under lattice translations, threefold rotation, and reflection.  To make this precise, they satisfy the following symmetry relations:
\begin{align}
 U(z+\gamma)&=\mathrm{e}^{\mathrm{i}\ip{\gamma}{K}}U(z), \quad \forall \gamma\in\Lambda&
 U(\omega z)&=\omega U(z),&
 \overline{U(\bar z)}&=-U(-z),\label{eq:Usym}\\
 V(z+\gamma)&=\mathrm{e}^{\mathrm{i}\ip{\gamma}{K}}V(z), \quad \forall \gamma\in\Lambda&
 V(\omega z)&=V(z),&
 V(\bar z)&=V(z)=\overline{V(-z)}.\label{eq:Vsym}
\end{align}
{ We call an equivariant nonchiral tunnelling potential
\(V\in L^\infty_{\rm loc}(\C)\) satisfying \eqref{eq:Vsym} \emph{admissible}.
Throughout this paper, measurable functions are identified up to equality almost
everywhere, and all identities involving them, including symmetry relations,
are understood in this sense.  Since \eqref{eq:Vsym} makes \(\lvert V\rvert\)
periodic, we equip the resulting real Banach space \(\mathcal V_{\rm adm}\)
with the norm \(\|V\|_{L^\infty(X)}\).}
A further simplification of the model, which is widely used in the literature, is to truncate the higher Fourier modes of the tunnelling potentials, and only keep the first-harmonic terms in $U$ and $V$:
\begin{equation}\label{eq:BMpotentials}
 \begin{aligned}
 U_{\rm BM}(z)&=-\frac{4\pi\mathrm{i}}{3}
 \sum_{\ell=0}^{2}\omega^\ell
 \mathrm{e}^{\mathrm{i}\ip{z}{\omega^\ell K}},\quad
 V_{\rm BM}(z)=\sum_{\ell=0}^{2}
 \mathrm{e}^{\mathrm{i}\ip{z}{\omega^\ell K}}.
 \end{aligned}
\end{equation}
Despite its wide use, such first-harmonic approximations neglect longer-range tunnelling effects, as discussed in \cite{QuinnKongLuskinWatson}. In this paper, we are interested in general $U$ and $V$ satisfying the required symmetries.

The relation \(V(-z)=\overline{V(z)}\) makes \(W_V\) Hermitian.  To describe
the translation symmetry, for \(\gamma\in\Lambda\) let
\[
 T_\gamma=
 \begin{pmatrix}
  \mathrm{e}^{-\mathrm{i}\ip{\gamma}{K}}&0\\
 0&\mathrm{e}^{\mathrm{i}\ip{\gamma}{K}}
 \end{pmatrix}.
\]
Set \(\mathbf T_\gamma=\operatorname{diag}(T_\gamma,T_\gamma)\) and define
the combined lattice translation by
\((\mathcal L_\gamma\Psi)(z)=\mathbf T_\gamma\Psi(z-\gamma)\).
The translation laws \eqref{eq:Usym}--\eqref{eq:Vsym} imply
\(H_{\rm BM}(\alpha,\lambda)\mathcal L_\gamma
=\mathcal L_\gamma H_{\rm BM}(\alpha,\lambda)\).
Hence Bloch--Floquet
decomposition applies, and it is enough to consider the resulting operators
\(H_k\).  Let \(L^2_0\) and \(H^1_0\) denote, respectively, the spaces of functions
\(f\in L^2_{\rm loc}(\C;\C^2)\) and
\(f\in H^1_{\rm loc}(\C;\C^2)\) satisfying
\begin{equation}\label{eq:twisted-boundary}
 f(z+\gamma)=T_\gamma f(z),\quad
 \gamma\in\Lambda,\quad\text{for a.e. }z\in\C.
\end{equation}
Since \(T_\gamma\) is unitary, the corresponding norms are well-defined on \(X\).  Set
\[
 \mathscr H_0=L^2_0\oplus L^2_0,\quad
 \mathscr D_0=H^1_0\oplus H^1_0.
\]
Define the reciprocal lattice by
\[
 \Lambda^*=\{G\in\C:\ip{\gamma}{G}\in2\pi\Z
 \text{ for every }\gamma\in\Lambda\}.
\]
 The Bloch Hamiltonian at momentum
\(k\), denoted as $H_k$, is given by
\begin{equation}\label{eq:H}
 H_k(\alpha,\lambda)=
 \begin{pmatrix}
  \lambda W_V&D(\alpha)^*+\bar k\\
  D(\alpha)+k&\lambda W_V
 \end{pmatrix}:\mathscr D_0\longrightarrow\mathscr H_0.
\end{equation}
Note that the operators \(H_k\) are indexed by \(k\in\C/\Lambda^*\).  For real \(\alpha\) and \(\lambda\), \(H_k(\alpha,\lambda)\) is self-adjoint with compact resolvent.  Moreover, \(H_{k+G}\) is unitarily equivalent to \(H_k\) for every \(G\in\Lambda^*\).  Thus its discrete eigenvalues $E_n(k)$, which could  be viewed as functions of $k$, define energy bands on \(\C/\Lambda^*\).  A band $E_n(k)$ is called exactly flat if its value is independent of \(k\).  Moreover, we say that the model has at least \(m\) zero-energy flat bands if
\(\dim\ker H_k\geq m\) for every \(k \in \C/\Lambda^*\).

At a magic value of \(\alpha\), the chiral operator, namely $(D(\alpha) + k)$, has a zero mode at every momentum $k$.  We call a real parameter \(\alpha_*\) \emph{simple magic} if this zero mode is unique up to multiplication by a scalar, that is, if
\[
 \dim_\C\ker_{L^2_0}(D(\alpha_*)+k)=1,\quad\forall k\in\C.
\]
This is the minimal-multiplicity condition of
\cite[Theorem~2]{BHZfine}.  For the standard potential \(U=U_{\rm BM}\), the
existence of the first real magic parameter was proved in
\cite{WatsonLuskin}, and its simplicity was proved in
\cite[Theorem~3]{BHZintegrability}.     The following lemma collects the properties of a
simple chiral zero mode used in this work.

\begin{lemma}[Structure of a simple chiral zero mode]
\label{lem:chiral-zero-mode}
Let \(U\) be a nonzero real-analytic tunnelling potential satisfying \eqref{eq:Usym}, and
let \(\alpha_*\) be a simple real magic parameter for this potential.  Then the
following statements hold.
\begin{enumerate}
\item For every \(k\in\C\), both
\(\ker(D(\alpha_*)+k)\) and
\(\ker(D(\alpha_*)^*+\bar k)\) are one-dimensional.  Consequently,
\(H_k(\alpha_*,0)\) has exactly two zero modes, and these eigenvalues are
uniformly separated from the rest of the spectrum.
\item An \(L^2\)-normalized zero mode
\(\vec u_0\in\ker D(\alpha_*)\) can be chosen in the form
\begin{equation}\label{eq:u0}
 \vec u_0(z)=
 \begin{pmatrix}p(z)\\ \eps\mathrm{i}p(-z)\end{pmatrix},
 \quad \eps\in\{+1,-1\}.
\end{equation}
Here \(p(z)\) satisfies that
\begin{equation}\label{eq:p-sym}
 p(z+\gamma)=\mathrm{e}^{-\mathrm{i}\ip{\gamma}{K}}p(z),
 \quad p(\omega z)=\omega p(z).
\end{equation}
The function \(p\) is real analytic.  Moreover, there is a constant
\(\rho\in\C\), \(|\rho|=1\), such that
\begin{equation}\label{eq:u0-reflection}
 \overline{p(-\bar z)}=\rho p(z),
 \quad
 \overline{p(\bar z)}=\rho p(-z).
\end{equation}
\item The vector \(\vec u_0\) has a unique zero on \(X\), at \([0]\), and
this zero has order one. 
Thus, on a coordinate disc centered at \(0\),
\begin{equation}\label{eq:u0-local-factor}
 \vec u_0(z)=z\vec w(z),\quad \vec w(0)\neq0,
\end{equation}
for a smooth \(\C^2\)-valued function \(\vec w\).
\end{enumerate}
\end{lemma}

For item~(1), see \cite[Theorem~2]{BHZfine}.  For item~(2), see
\cite[Eqs.~(2.7)--(2.8) and Sec.~4.3]{BeckerZworski}; the real analyticity of
\(p\) follows from analytic elliptic regularity applied to
\(D(\alpha_*)\vec u_0=0\).  For item~(3), see \cite[Theorem~3]{BHZfine} and
also \cite[Proposition~3.6]{BeckerZworskiDirac}.

% \vspace{1cm}

Now we are ready to state our main results. 
We start with the following definitions:
\begin{definition}
Let $p(z)$ be as in Lemma~\ref{lem:chiral-zero-mode} and define
\begin{equation}\label{eq:V0}
 V_0(z)=p(z)\overline{p(-z)}.
\end{equation}
Let \(\mathcal F_{\rm sym}\) be the real vector space of measurable
functions \(f:\C\to\R\) such that \(fV_0\in L^\infty(X)\) and
\begin{equation}\label{eq:f-sym}
 \begin{aligned}
 f(z+\gamma)&=f(z),& f(\omega z)&=f(z),\\
 f(\bar z)&=f(z),& f(-z)&=f(z),
 \end{aligned}
 \quad \gamma\in\Lambda{ .}
\end{equation}
For \(f\in\mathcal F_{\rm sym}\), set
\begin{equation}\label{eq:Vf}
 V[f](z)=f(z)V_0(z)=f(z)p(z)\overline{p(-z)}.
\end{equation}
\end{definition}

\begin{theorem}[Exact flat bands beyond the chiral limit]
\label{thm:exact}
Let \(U\) be a nonzero real-analytic tunnelling potential satisfying
\eqref{eq:Usym}, let \(\alpha_*\) be a simple real magic parameter for this
potential, choose \(f\in\mathcal F_{\rm sym}\setminus\{0\}\), and set
\(V=V[f]\) as in \eqref{eq:Vf}.  Then
\(U,V[f]\not\equiv0\), the symmetries \eqref{eq:Usym}--\eqref{eq:Vsym} hold, and
\[\dim_\C \ker_{\mathscr H_0}H_k(\alpha_*,\lambda)\geq 2 \quad(k\in\C,\ \lambda\in\R).
\]
Thus every nonzero real nonchiral coupling strength \(\lambda\) gives at least two exact zero-energy flat bands.
\end{theorem}
Theorem~\ref{thm:characterization} shows that this
construction is complete: every admissible \(V\) with at least two
zero-energy states for every \(k\) and every real \(\lambda\) is of the form
\(V=V[f]\) for some \(f\in\mathcal F_{\rm sym}\).
In particular, take \(U=U_{\rm BM}\), for which a simple real magic parameter
is known to exist.  The theorem then provides a family of counterexamples to Open
Problem~3 of \cite{ZworskiSurvey}. 

We next give a complementary obstruction result to exact flatness.  At
\(\lambda=0\), the Hamiltonian has exactly two zero-energy states at each
momentum.  Denote their span by
\[
 \mathcal Z_k:=\ker H_k(\alpha_*,0),
\]
and let \(P_k\) be the orthogonal projection onto \(\mathcal Z_k\).  The
compression of the nonchiral perturbation to this two-dimensional space is
\begin{equation}\label{eq:compression}
 A_V(k):=P_k\mathbf W_VP_k:\mathcal Z_k\longrightarrow\mathcal Z_k.
\end{equation}
We refer to \(A_V(k)\) as the compressed perturbation.  The symmetries force
\(A_V(k)=a_V(k)I_2\), as shown in Sec.~\ref{sec:nonflatness}; hence
invertibility amounts to the single scalar condition \(a_V(k)\neq0\).

\medskip
\noindent\textbf{Invertibility criterion.}
For an admissible nonchiral \textcolor{black}{tunnelling} potential \(V\), the following two
conditions are equivalent (see \eqref{eq:vanishing-compression}):
\begin{enumerate}
\item there exists \(k_0\in\C/\Lambda^*\) such that
\begin{equation}\label{eq:invertibility-criterion}
 A_V(k_0)\text{ is invertible},\quad\text{or equivalently,}\quad
 a_V(k_0)\neq0;
\end{equation}
\item \(\operatorname{Im}\!\left(V(z)p(-z)\overline{p(z)}\right)\) is
nonzero on a set of positive measure in \(X\).
\end{enumerate}
We say that \(V\) satisfies the invertibility criterion if either of these
conditions holds, and hence if both do.

\begin{proposition}[Nonflatness under the invertibility criterion]\label{prop:compression}
Let \(U\) be a nonzero real-analytic tunnelling potential satisfying
\eqref{eq:Usym}, let \(\alpha_*\) be a simple real magic parameter for this
potential, and let \(V\) be admissible.  Suppose that \(V\) satisfies the
invertibility criterion, and let \(k_0\in\C/\Lambda^*\) be a witnessing
momentum.  Then there exists \(\delta>0\) such that
\[
 H_{k_0}(\alpha_*,\lambda)\text{ has no zero eigenvalue} \quad \text{for } 0<|\lambda|<\delta.
\]
Consequently, for every sufficiently small nonzero \(\lambda\), neither of
the two central bands is constant in \(k\).
\end{proposition}

We prove Proposition~\ref{prop:compression} below by a direct linear-algebraic
argument.  It is also an immediate consequence of the perturbative
eigenvalue expansion in \cite[Eq.~(1.2) and Lemma~4]{BeckerZworski}.
 
The same criterion gives a global conclusion in the coupling parameter.

\begin{theorem}[Nonflatness outside a discrete set of couplings]
\label{thm:discrete-couplings}
Let \(U\) be a nonzero real-analytic tunnelling potential satisfying
\eqref{eq:Usym}, let \(\alpha_*\) be a simple real magic parameter for this
potential, and let \(V\) be admissible.  Suppose that \(V\) satisfies the
invertibility criterion, and let \(k_0\) be a witnessing momentum.  Then
there is a discrete set \(\Sigma_V\subset\R\), with no finite accumulation
point, such that the two central bands are not flat for every
\(\lambda\in\R\setminus\Sigma_V\).  In particular, every bounded interval
contains only finitely many coupling strengths at which exact flatness can
occur.
\end{theorem}

Proposition~\ref{prop:compression} excludes exact flatness throughout a
punctured neighborhood of \(\lambda=0\), whereas
Theorem~\ref{thm:discrete-couplings} applies over the full real coupling
axis, allowing only isolated exceptional values.

Finally, we apply the invertibility criterion to the standard BM model, where
\(U=U_{\rm BM}\) and \(V=V_{\rm BM}\) are given by
\eqref{eq:BMpotentials}.

\begin{proposition}[Nonvanishing of the standard BM compression]
\label{prop:standard-BM}
Let \(\alpha_*\) be a simple real magic parameter and take \(U=U_{\rm BM}\) and \(V=V_{\rm BM}\) from
\eqref{eq:BMpotentials}.  Then
\[
 a_{V_{\rm BM}}\not\equiv0\quad\text{on }\C/\Lambda^*.
\]
Equivalently, \(V_{\rm BM}\) satisfies the invertibility criterion. 
\end{proposition}

Proposition~\ref{prop:standard-BM} provides a momentum \(k_0\) for which
\(A_{V_{\rm BM}}(k_0)\) is invertible.  Proposition~\ref{prop:compression}
therefore excludes exact flatness for every sufficiently small nonzero real
\(\lambda\), while Theorem~\ref{thm:discrete-couplings} shows that, over the
full real coupling axis, exact flatness can occur only at a discrete set of
values with no finite accumulation point.

\section{Exact flat bands beyond the chiral limit}\label{sec:construction}

We first check that the potentials \(V[f]\), defined in \eqref{eq:Vf} for
\(f\in\mathcal F_{\rm sym}\), are admissible.  By
Lemma~\ref{lem:chiral-zero-mode}, the
first component of the chiral zero mode satisfies \eqref{eq:p-sym}.
Using \(V_0(z)=p(z)\overline{p(-z)}\), these identities give
\(V_0(z+\gamma)=\mathrm{e}^{-2\mathrm{i}\ip{\gamma}{K}}V_0(z)
               =\mathrm{e}^{\mathrm{i}\ip{\gamma}{K}}V_0(z),\)
because \(\mathrm{e}^{3\mathrm{i}\ip{\gamma}{K}}=1\), and that
\(V_0(\omega z)=V_0(z)\). We also have
\(V_0(-z)=\overline{V_0(z)}\).

It remains to check reflection.  Substitution of \eqref{eq:u0} into the
reflection relations \eqref{eq:u0-reflection} gives
\[
 \overline{p(-\bar z)}=\rho p(z),
 \quad p(\bar z)=\bar\rho\,\overline{p(-z)},
\]
and hence \(V_0(\bar z)=V_0(z)\).  Thus \eqref{eq:Vsym} holds.  By
Lemma~\ref{lem:chiral-zero-mode}, \(p\) is real analytic.  Since $p(z)\neq0$ for some $z$, the nonzero sets of $p(z)$ and $\overline{p(-z)}$ are open and
dense; therefore \(V_0(z)\not\equiv0\).
\begin{figure}[H]
\centering
\includegraphics[width=0.96\textwidth]{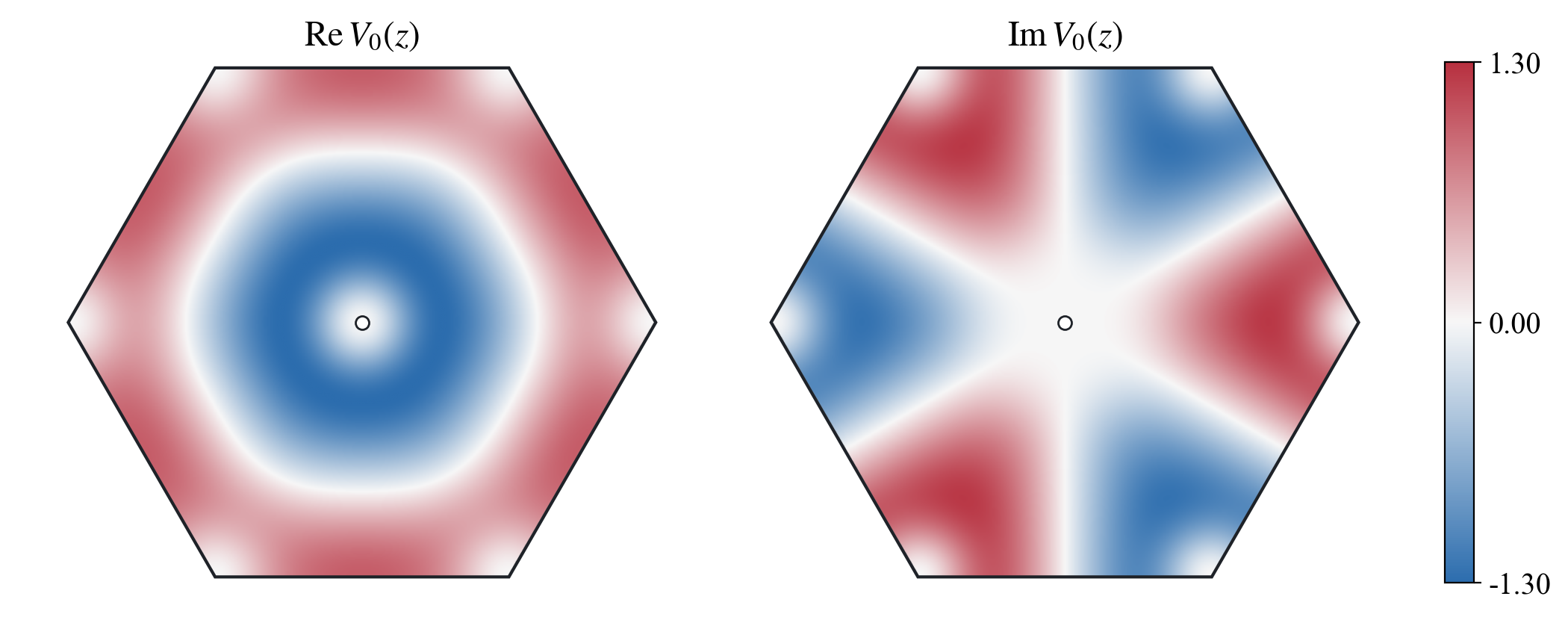}
\caption{Real and imaginary parts of \(V_0\) on a fundamental cell of its
period lattice.  The white circle marks the zero at \(z=0\).}
\label{fig:v0}
\end{figure}

Multiplication by an amplitude \(f\) satisfying
\eqref{eq:f-sym} preserves these transformation laws, and
\(fV_0\in L^\infty(X)\) by definition.  Thus every \(V[f]=fV_0\) with
nonzero \(f\in\mathcal F_{\rm sym}\) is admissible.

We now fix one such potential and prove that its Hamiltonian has at least two
zero modes for every Bloch momentum \(k\).
The twisted boundary condition \eqref{eq:twisted-boundary} defines the
rank-two bundle
\[
 E_0=(\C\times\C^2)/{\sim}\longrightarrow X,
 \quad (z,\xi)\sim(z+\gamma,T_\gamma\xi).
\]
We first construct two line subbundles \(L,M\subset E_0\) that define a
closed subsystem for the Hamiltonian, and then compute the index of its
restriction.

The chiral Hamiltonian has one zero mode from \(D(\alpha_*)\) and one from its
adjoint.  The anti-linear operation
\[
 (\mathscr Qf)(z)=\overline{f(-z)}
\]
turns the first into the second.  Indeed,
\(\mathscr QD(\alpha_*)\mathscr Q=D(\alpha_*)^*\), so
\begin{equation}\label{eq:v0}
 \vec v_0(z):=(\mathscr Q\vec u_0)(z)
 =\begin{pmatrix}
   \overline{p(-z)}\\-\eps\mathrm{i}\overline{p(z)}
  \end{pmatrix}
 \in\ker D(\alpha_*)^*.
\end{equation}
At each point where they are nonzero, \(\vec u_0(z)\) and \(\vec v_0(z)\)
span orthogonal lines in \(\C^2\):
\[
 \vec u_0(z)^*\vec v_0(z)
 =\overline{p(z)p(-z)}
  +(-\eps\mathrm{i}\overline{p(-z)})
   (-\eps\mathrm{i}\overline{p(z)})=0.
\]
Fix \(f\in\mathcal F_{\rm sym}\), and define
\[
 W(z):=W_{V[f]}(z)=f(z)W_{V_0}(z),
 \quad \mu_f(z):=\eps\mathrm{i}f(z)p(z)p(-z).
\]
Since \(|\mu_f|=|V[f]|\), the matrix coefficients of \(W\), as well as
\(\mu_f\), belong to \(L^\infty(X)\).
The special form of \(V[f]\) makes \(W(z)\) map each zero-mode direction into
the other.  Direct calculation gives
\begin{equation}\label{eq:line-mapping}
 W(z)\vec u_0(z)=\mu_f(z)\vec v_0(z),
 \quad W(z)\vec v_0(z)=\overline{\mu_f(z)}\vec u_0(z)
 \quad\text{a.e.}
\end{equation}
For example,
\[
 W(z)\vec u_0(z)=
 \begin{pmatrix}
  \eps\mathrm{i}f(z)p(z)|p(-z)|^2\\f(z)p(-z)|p(z)|^2
 \end{pmatrix}
 =\eps\mathrm{i}f(z)p(z)p(-z)
 \begin{pmatrix}
  \overline{p(-z)}\\-\eps\mathrm{i}\overline{p(z)}
 \end{pmatrix}.
\]
We now show that, although \(\vec u_0\) and \(\vec v_0\) vanish at \(z=0\),
the lines they span extend smoothly across \(z=0\).  Away from zero, each
nonzero vector \(\vec u_0(z)\) or \(\vec v_0(z)\) selects a line in the
corresponding fiber of \(E_0\).  Indeed,
Equation~\eqref{eq:u0-local-factor} in
Lemma~\ref{lem:chiral-zero-mode} and the definition of \(\vec v_0\) in
Equation~\eqref{eq:v0} imply that, near \(z=0\),
\begin{equation}\label{eq:local-factor}
 \vec u_0(z)=z\vec w(z),\quad
 \vec v_0(z)=\bar z\vec w_M(z),
\end{equation}
where
\(
 \vec w_M(z):=-\overline{\vec w(-z)}
\).
In particular, \(\vec w(0),\vec w_M(0)\neq0\).
We therefore define \(L\) and \(M\) as follows:
\begin{equation}\label{eq:fibres}
 L_{[z]}=
 \begin{cases}
  \C\vec u_0(z),&{ [z]\neq[0]},\\
  \C\vec w(0),&{ [z]=[0]},
 \end{cases}
 \quad
 M_{[z]}=
 \begin{cases}
  \C\vec v_0(z),&{ [z]\neq[0]},\\
  \C\vec w_M(0),&{ [z]=[0]}.
 \end{cases}
\end{equation}
Here \(\C\vec u_0(z)=\{c\vec u_0(z):c\in\C\}\) is a one-dimensional subspace of \(\C^2\).  Since both \(\vec u_0\) and \(\vec v_0\) satisfy the twisted boundary condition \eqref{eq:twisted-boundary},  \(L\) and \(M\) are line subbundles inside \(E_0\). The pointwise orthogonality of \(\vec u_0\) and \(\vec v_0\) extends through zero by continuity, so these line bundles give a smooth orthogonal splitting \(E_0=L\oplus M\).

The next lemma records how the four blocks of the BM Hamiltonian act on the
pair \(L,M\).

\begin{lemma}[Mapping properties on \(L\) and \(M\)]\label{lem:closed}
\begin{align*}
 D(\alpha_*)&:C^\infty(\mathbb{C} / \Lambda; L)\longrightarrow C^\infty(\mathbb{C} / \Lambda; L),\quad 
 D(\alpha_*)^*:C^\infty(\mathbb{C} / \Lambda; M)\longrightarrow C^\infty(\mathbb{C} / \Lambda; M).
\end{align*}
More precisely, the differential restrictions
\(D(\alpha_*)|_L:H^1(\mathbb{C} / \Lambda;L)\to L^2(\mathbb{C} / \Lambda; L)\) and
\(D(\alpha_*)^*|_M:H^1(\mathbb{C} / \Lambda; M)\to L^2(\mathbb{C} / \Lambda; M)\) are bounded.  Multiplication by \(W\)
maps \(L\) to \(M\) and \(M\) to \(L\) and induces
bounded operators
\[
 W|_L:L^2(\mathbb{C} / \Lambda; L)\longrightarrow L^2(\mathbb{C} / \Lambda; M),\quad
 W|_M:L^2(\mathbb{C} / \Lambda; M)\longrightarrow L^2(\mathbb{C} / \Lambda; L).
\]
\end{lemma}

\begin{proof}
For $z\neq 0$, the nonvanishing vectors \(\vec u_0\) and
\(\vec v_0\) are local frames for \(L\) and \(M\), respectively.  Since
\(D(\alpha_*)\vec u_0=0\) and \(D(\alpha_*)^*\vec v_0=0\), the Leibniz rule
gives, for local scalar functions \(f\) and \(g\),
\[
 D(\alpha_*)\bigl(f\vec u_0\bigr)=2(D_{\bar z}f)\vec u_0,
 \quad
 D(\alpha_*)^*\bigl(g\vec v_0\bigr)=2(D_zg)\vec v_0.
\]

Near \(z=0\), the nonvanishing vectors \(\vec w\) and \(\vec w_M\) are
local frames for \(L\) and \(M\).  For \(z\neq0\) near \(0\), the
identities \(\vec u_0=z\vec w\) and \(\vec v_0=\bar z\vec w_M\), together
with \(D_{\bar z}z=D_z\bar z=0\), give
\(D(\alpha_*)\vec w=0\) and \(D(\alpha_*)^*\vec w_M=0\).  By smoothness,
these identities also hold at \(z=0\).  Hence
\[
 D(\alpha_*)\bigl(f\vec w\bigr)=2(D_{\bar z}f)\vec w,
 \quad
 D(\alpha_*)^*\bigl(g\vec w_M\bigr)=2(D_zg)\vec w_M.
\]
These two sets of local frames cover \(X\), so $D(\alpha_*)$ and \(D(\alpha_*)^*\) preserve \(L\) and \(M\) globally.

Equation~\eqref{eq:line-mapping} shows that \(W\) maps \(L\) to \(M\) and
\(M\) to \(L\); the single point \([0]\), where the
frames \(\vec u_0\) and \(\vec v_0\) vanish, is irrelevant for \(L^2\)
sections.  Since the matrix coefficients of \(W\) belong to \(L^\infty(X)\),
these restrictions are bounded on \(L^2\).  The first-order differential
restrictions are bounded from \(H^1\) to \(L^2\) on the compact torus \(X\).
\end{proof}

We now use the line bundles \(L\) and \(M\) to define a restricted Hamiltonian
and count its zero modes using the Fredholm index.  For a Fredholm operator
\(T\),
\[
 \ind_\C T=\dim_\C\ker T-\dim_\C\operatorname{coker}T,
\]
so \(\dim_\C\ker T\geq\ind_\C T\).  The restricted differential operators
considered below are elliptic on the compact real-space torus
\(X=\C/\Lambda\), and hence Fredholm.

\begin{lemma}\label{lem:line-indices}
The line bundles \(L\) and \(M\) satisfy
\[
 \deg_XL=1,\quad \deg_XM=-1.
\]
For every \(k\in\C\), the restricted operators
\((D(\alpha_*)+k)|_L\) and \((D(\alpha_*)^*+\bar k)|_M\) are Fredholm, with
\begin{equation}\label{eq:line-indices}
 \ind_\C(D(\alpha_*)+k)|_L=1,\quad
 \ind_\C(D(\alpha_*)^*+\bar k)|_M=1.
\end{equation}
\end{lemma}

\begin{proof}
By Lemma~\ref{lem:chiral-zero-mode}, \(\vec u_0\) is a global section of
\(L\) with a unique simple zero at \([0]\).  It therefore trivializes \(L\)
on \(X\setminus\{[0]\}\).  In the nonvanishing local frame \(\vec w\) near
zero, it has coefficient \(z\), whose winding number is \(+1\).  Since the
degree is the sum of the local winding numbers of the zeros of a section,
\(\deg_XL=1\).  Likewise, \(\deg_XM=-1\).

On \(L\), \(D(\alpha_*)\) differentiates the scalar coefficient by
\(2D_{\bar z}\), so \((D(\alpha_*)+k)|_L\) is a Cauchy--Riemann operator on
\(L\), plus a zeroth-order multiplication term.  The Riemann--Roch theorem
states that a Cauchy--Riemann operator on a line bundle \(E\) over a compact
genus-\(g\) surface has complex Fredholm index
\[
 \deg_XE+1-g.
\]
Multiplication by a nonzero constant and addition of a zeroth-order term do
not change the index.  Since \(X\) is a torus, \(g=1\), and therefore the
first identity in \eqref{eq:line-indices} follows from \(\deg_XL=1\).

The operator on \(M\) uses \(2D_z\), the anti-Cauchy--Riemann operator, whose
complex index is minus the degree.  Hence the second identity in
\eqref{eq:line-indices} follows from \(-\deg_XM=1\).
\end{proof}

We can now restrict the four-component Hamiltonian to sections whose first
two-component entry lies in \(L\) and whose second lies in \(M\):
\begin{equation}\label{eq:Hres}
 \calH^{\rm res}_{k,\lambda}:
 H^1(\mathbb{C} / \Lambda; L)\oplus H^1(\mathbb{C} / \Lambda; M)\longrightarrow L^2(\mathbb{C} / \Lambda; M)\oplus L^2(\mathbb{C} / \Lambda; L),
\end{equation}
where
\[
 \calH^{\rm res}_{k,\lambda}=
 \begin{pmatrix}
  \lambda W|_L&(D(\alpha_*)^*+\bar k)|_M\\
  (D(\alpha_*)+k)|_L&\lambda W|_M
 \end{pmatrix}.
\]
\begin{lemma}\label{lem:restricted-index}
For every \(k\in\C\) and \(\lambda\in\R\), the operator
\(\calH^{\rm res}_{k,\lambda}\) is Fredholm and
\[
 \ind_\C\calH^{\rm res}_{k,\lambda}=2.
\]
\end{lemma}

\begin{proof}
Lemma~\ref{lem:closed} shows that $\calH_{k,\lambda}^{\rm res}$ is well defined.  At \(\lambda=0\), \(\calH^{\rm res}_{k,0}\) is the direct sum of the two Fredholm operators whose indices are given in \eqref{eq:line-indices}.  Therefore \(\calH^{\rm res}_{k,0}\) is Fredholm, and $\ind_\C\calH^{\rm res}_{k,0}=2$.
By Lemma~\ref{lem:closed},
\[
 \calH^{\rm res}_{k,\lambda}-\calH^{\rm res}_{k,0}
 =\lambda
 \begin{pmatrix}
  W|_L&0\\
  0&W|_M
 \end{pmatrix}
\]
is compact.  Indeed, each diagonal block factors as the compact inclusion
\(H^1(\mathbb{C} / \Lambda; L)\hookrightarrow L^2(\mathbb{C} / \Lambda; L)\) or
\(H^1(\mathbb{C} / \Lambda; M)\hookrightarrow L^2(\mathbb{C} / \Lambda; M)\), followed by the corresponding bounded
\(L^2\) multiplication operator from Lemma~\ref{lem:closed}.
% Therefore, by Atkinson's theorem, $\ind_\C\calH^{\rm res}_{k,\lambda} =\ind_\C\calH^{\rm res}_{k,0}=2$.
Therefore $\ind_\C\calH^{\rm res}_{k,\lambda} =\ind_\C\calH^{\rm res}_{k,0}=2$, since compact perturbations do not change the Fredholm index.
\end{proof}

\begin{proof}[Proof of Theorem~\ref{thm:exact}]
By Lemma~\ref{lem:restricted-index}, $\dim_\C\ker\calH^{\rm res}_{k,\lambda} \geq\ind_\C\calH^{\rm res}_{k,\lambda}=2$.
Every restricted zero mode is also a zero mode of the full operator
\eqref{eq:H}, which proves the theorem.
\end{proof}

\begin{remark}\label{rem:first-order-cancellation}
The first-order cancellation can also be seen directly from the perturbative
formula of Becker--Zworski.  In our notation, their expansion of the two
central-band energies is
\begin{equation}\label{eq:first-order-band-expansion}
 E_m(\alpha_*,\lambda,k)
 =\lambda a_V(k)+\mathcal O(\lambda^2),
 \quad m\in\{-1,1\},\quad \lambda\to0.
\end{equation}
The coefficient \(a_V(k)\) here is the same scalar defined later by
\(A_V(k)=a_V(k)I_2\) in \eqref{eq:scalarA}.  In terms of the scalar Bloch
multiplier \(F_k\) and the normalizing constant \(\nu_k>0\) introduced in
Lemma~\ref{lem:theta-zero-modes}, it takes the form
\[
 a_V(k)=-2\eps\nu_k^2\int_X |F_k(z)|^2
 \operatorname{Im}\!\left(V(z)p(-z)\overline{p(z)}\right)\,\mathrm{d}m(z),
\]
as derived in \eqref{eq:aformula} below; see also
\cite[Eq.~(3.9)]{BeckerZworski}.  Here \(F_k\) is defined explicitly in
\eqref{eq:Fk}, and we set \(\mathrm{d}m(z)=dx\,dy\).
For \(V=V[f]\), the expression inside the imaginary part is
\(f(z)|p(z)p(-z)|^2\), which is real.  Hence
\(a_V(k)=0\) for every \(k\).
This provides a useful first-order consistency check, while
Theorem~\ref{thm:exact} gives the stronger, nonperturbative conclusion that
two zero-energy flat bands persist for every real \(\lambda\).
Conversely, if two zero-energy flat bands persist for every
real \(\lambda\), then \eqref{eq:first-order-band-expansion} forces
\(a_V(k)=0\) for every \(k\).  Theorem~\ref{thm:characterization} makes this
condition precise: the admissible potentials satisfying \(a_V(k)=0\) for
every \(k\) are exactly those of the form \(V=V[f]\) with
\(f\in\mathcal F_{\rm sym}\).
\end{remark}

\section{Invertibility criterion and nonexact flatness}
\label{sec:nonflatness}

Recall that we have defined
\[
 \mathcal Z_k=\ker H_k(\alpha_*,0),
 \quad
 P_k:\mathscr H_0\longrightarrow\mathcal Z_k
\]
as the two-dimensional chiral zero-mode space and the orthogonal projection
onto it.  Set \(Q_k=I-P_k\).  Recall also that
\(A_V(k)=P_k\mathbf W_VP_k\) is a linear operator on \(\mathcal Z_k\).

For each \(k\), choose normalized zero modes
\[
 \vec u_k\in\ker(D(\alpha_*)+k),
 \quad
 \vec v_k\in\ker(D(\alpha_*)^*+\bar k).
\]
Then \((\vec u_k,\boldsymbol 0)\) and \((\boldsymbol 0,\vec v_k)\) form an orthonormal basis of
\(\mathcal Z_k\).  Because \(\mathbf W_V\) is block diagonal,
\(A_V(k)\) is diagonal in this basis.  The real-model symmetries
imply that its two diagonal entries are equal
\cite[Lemma~4]{BeckerZworski}; hence
\begin{equation}\label{eq:scalarA}
 A_V(k)=a_V(k)I_2,
 \quad
 a_V(k)=\ip{\vec u_k}{W_V\vec u_k}_{L^2_0}.
\end{equation}
This is the same scalar \(a_V(k)\) that appears as the common first-order
coefficient in the perturbative expansion
\eqref{eq:first-order-band-expansion} of
Remark~\ref{rem:first-order-cancellation}.   

The invertibility criterion, as stated in
\eqref{eq:invertibility-criterion}, assumes that \(A_V(k_0)\) is invertible,
or equivalently \(a_V(k_0)\neq0\), for some \(k_0\). Now we prove that this
criterion implies nonexact flatness of the two central bands for small
nonchiral couplings \(\lambda\).

\begin{proof}[Proof of Proposition~\ref{prop:compression}]
Recall that \(\mathscr H_0=L^2_0\oplus L^2_0\) is the Bloch-fiber Hilbert
space and \(\mathscr D_0=H^1_0\oplus H^1_0\) is the domain of
\(H_{k_0}(\alpha_*,\lambda)\).
Since \(H_{k_0}(\alpha_*,0)\) is self-adjoint with compact resolvent, its
restriction
\[
 B_{k_0}(0):=Q_{k_0}H_{k_0}(\alpha_*,0)Q_{k_0}:
 Q_{k_0}\mathscr D_0\longrightarrow Q_{k_0}\mathscr H_0
\]
is invertible.  Because \(\mathbf W_V\) is bounded, the perturbed block
\[
 B_{k_0}(\lambda):=Q_{k_0}H_{k_0}(\alpha_*,\lambda)Q_{k_0}
 =B_{k_0}(0)+\lambda Q_{k_0}\mathbf W_VQ_{k_0}
\]
remains invertible for small \(|\lambda|\); write
\(R_{k_0}(\lambda)=B_{k_0}(\lambda)^{-1}\).  These inverses are uniformly
bounded near \(\lambda=0\).

With respect to the decompositions
\[
 \mathscr D_0=\mathcal Z_{k_0}\oplus Q_{k_0}\mathscr D_0,
 \quad
 \mathscr H_0=\mathcal Z_{k_0}\oplus Q_{k_0}\mathscr H_0,
\]
the Schur complement of \(B_{k_0}(\lambda)\) is
\begin{equation}\label{eq:Feshbach}
 S_{V,k_0}(\lambda)
 =\lambda A_V(k_0)
 -\lambda^2P_{k_0}\mathbf W_VQ_{k_0}R_{k_0}(\lambda)
 Q_{k_0}\mathbf W_VP_{k_0}.
\end{equation}
The Schur-complement criterion shows that
\(H_{k_0}(\alpha_*,\lambda)\) is invertible if and only if
\(S_{V,k_0}(\lambda)\) is invertible.  For \(\lambda\neq0\),
\[
 \lambda^{-1}S_{V,k_0}(\lambda)
 =A_V(k_0)-\lambda P_{k_0}\mathbf W_VQ_{k_0}R_{k_0}(\lambda)
 Q_{k_0}\mathbf W_VP_{k_0}.
\]
As $\lambda \rightarrow 0$, the right-hand side converges in operator norm to the invertible operator
\(A_V(k_0)\).  It is therefore invertible for all sufficiently small
\(|\lambda|\), proving that \(H_{k_0}(\alpha_*,\lambda)\) has no zero
eigenvalue when \(0<|\lambda|<\delta\).

It remains to turn the absence of a zero mode at \(k_0\) into nonflatness.
At the Dirac momenta, the protected-state theorem
\cite[Proposition~2]{BeckerZworski} gives
\[
 \dim\ker H_{\pm K}(\alpha_*,\lambda)\ge2
 \quad(\lambda\in\R).
\]
Thus the central bands pass through zero at \(\pm K\) but not at \(k_0\).
They cannot be flat.
\end{proof}

The perturbative expansion in Remark~\ref{rem:first-order-cancellation},
Eq.~\eqref{eq:first-order-band-expansion}, is given by the main theorem of
\cite[Eq.~(1.2)]{BeckerZworski} and yields the same conclusion directly.
Since
\(A_V(k_0)=a_V(k_0)I_2\), invertibility of \(A_V(k_0)\) means that
\(a_V(k_0)\neq0\).  Equation~\eqref{eq:first-order-band-expansion} then gives
\[
 E_m(\alpha_*,\lambda,k_0)
 =\lambda a_V(k_0)+\mathcal O(\lambda^2)\neq0,
 \quad m\in\{-1,1\},
\]
for all sufficiently small nonzero \(\lambda\).  Together with the protected
zeros at \(k=\pm K\), this again proves nonflatness.

We next prove Theorem~\ref{thm:discrete-couplings}, which upgrades the local
perturbative obstruction in Proposition~\ref{prop:compression} to a
conclusion over the full coupling axis.

\begin{proof}[Proof of Theorem~\ref{thm:discrete-couplings}]
Proposition~\ref{prop:compression} provides a nonzero real number
\(\lambda_0\) for which \(H_{k_0}(\alpha_*,\lambda_0)\) is invertible.  For
complex \(\lambda\), the operators
\[
 H_{k_0}(\alpha_*,\lambda)
 =H_{k_0}(\alpha_*,0)+\lambda\mathbf W_V:
 \mathscr D_0\longrightarrow\mathscr H_0
\]
form an analytic Fredholm family.  Indeed,
\[
 H_{k_0}(\alpha_*,\lambda)H_{k_0}(\alpha_*,\lambda_0)^{-1}
 =I+(\lambda-\lambda_0)\mathbf W_V
 H_{k_0}(\alpha_*,\lambda_0)^{-1},
\]
and the operator after \(I\) is compact: the inverse maps into
\(\mathscr D_0\), whose inclusion in \(\mathscr H_0\) is compact, while
\(\mathbf W_V\) is bounded.  The analytic Fredholm theorem
\cite[Theorem~VI.14]{ReedSimonI} therefore shows that
\[
 \Sigma_V:=\{\lambda\in\R:
 H_{k_0}(\alpha_*,\lambda)\text{ is not invertible}\}
\]
is discrete and has no finite accumulation point.

As observed in the proof of Proposition~\ref{prop:compression}, exact
flatness of a central band would make
\(H_{k_0}(\alpha_*,\lambda)\) noninvertible.  Hence exact flatness can occur only when \(\lambda\in\Sigma_V\).
\end{proof}

We now apply the invertibility criterion to the standard  BM
model.  Although it retains only the leading Fourier harmonics of the
tunnelling potentials, it is the most widely used continuum model of twisted
bilayer graphene and a standard starting point for many-body calculations.
To verify the criterion, we need to evaluate the scalar \(a_V(k)\).  Recall that the normalized zero mode
at \(k=0\) is
\[
 \vec u_0(z)=
 \begin{pmatrix}p(z)\\ \eps\mathrm{i}p(-z)\end{pmatrix},
 \quad \eps\in\{+1,-1\}.
\]
To express the zero mode at momentum \(k\) in terms of the one at momentum
zero, we need a scalar multiplier that is compatible with the lattice
periodicity and shifts the equation from \(D(\alpha_*)\) to
\(D(\alpha_*)+k\).  The standard construction of such a multiplier on the
complex torus uses a Jacobi theta function; its apparent pole at \(0\) is
cancelled by the zero of \(\vec u_0\).  We recall only the formula needed
below.

Let \(\vartheta(z)=\vartheta_1(z\mid\omega)\) be the odd Jacobi theta
function, normalized by
\begin{align*}
 \vartheta(z+m)&=(-1)^m\vartheta(z),\\
 \vartheta(z+n\omega)&=(-1)^n
 \mathrm{e}^{-\pi\mathrm{i}n^2\omega-2\pi\mathrm{i}nz}\vartheta(z),
 \quad m,n\in\Z.
\end{align*}
Its zeros are simple and occur precisely at the points of \(\Lambda\).  Set
\begin{equation}\label{eq:Fk}
 z(k)=\frac{\sqrt3\,k}{4\pi\mathrm{i}},
 \quad
 F_k(z)=\mathrm{e}^{\frac{\mathrm{i}}2(z-\bar z)k}
 \frac{\vartheta(z-z(k))}{\vartheta(z)}.
\end{equation}

The following standard construction is proved in
\cite[Eqs.~(3.3)--(3.5) and the discussion following Lemma~3.2]{BHZfine};
see also \cite[Eqs.~(3.6)--(3.8)]{BeckerZworski}.  We include the short proof
for completeness and to match our conventions.

\begin{lemma}[Theta-function formula for the zero modes]
\label{lem:theta-zero-modes}
The function \(F_k\) in \eqref{eq:Fk} is \(\Lambda\)-periodic, and
\(F_k\vec u_0\) extends smoothly across \(z=0\) and satisfies
\[
 (D(\alpha_*)+k)(F_k\vec u_0)=0.
\]
Consequently, the normalized zero mode at momentum \(k\) may be chosen as
\[
 \vec u_k(z)=\nu_kF_k(z)\vec u_0(z),
 \quad \nu_k>0.
\]
\end{lemma}

\begin{proof}
The theta transformation laws make \(F_k\) periodic.  Its possible pole at
\(0\) is cancelled by the simple zero of \(\vec u_0\) from
Lemma~\ref{lem:chiral-zero-mode}, so \(F_k\vec u_0\) extends smoothly there.
Moreover, in the sense of distributions,
\[
 (2D_{\bar z}+k)F_k
 =2\pi\mathrm{i}\frac{\vartheta(z(k))}{\vartheta'(0)}\delta_0.
\]
Multiplication by \(\vec u_0\) annihilates the delta term because
\(\vec u_0(0)=0\).  The product rule and
\(D(\alpha_*)\vec u_0=0\) therefore give the claimed equation.  Since the
kernel of \(D(\alpha_*)+k\) is one-dimensional, normalization completes the
proof.  This is the standard theta-function construction of chiral zero
modes \cite[Sec.~3.1]{BHZfine}.
\end{proof}

Direct substitution of Lemma~\ref{lem:theta-zero-modes} into
\eqref{eq:scalarA} gives
\begin{equation}\label{eq:aformula}
 a_V(k)=-2\eps \nu_k^2\int_X |F_k(z)|^2
 \operatorname{Im}\!\left(V(z)p(-z)\overline{p(z)}\right)\,
 \mathrm{d}m(z).
\end{equation}
Since \(-2\eps\nu_k^2\neq0\), the matrix \(A_V(k)\) is invertible if and
only if the integral in \eqref{eq:aformula} is nonzero.  The next lemma
shows that if this integral vanishes for every \(k\), then
\(\operatorname{Im}\!\left(V(z)p(-z)\overline{p(z)}\right)=0\) almost
everywhere on \(X\).

\begin{lemma}[The theta averages determine the integrand]
\label{lem:theta-averages}
For an admissible nonchiral tunnelling potential \(V\), if \(a_V(k)=0\) for
every \(k\), then
\(\operatorname{Im}\!\left(V(z)p(-z)\overline{p(z)}\right)=0\) almost
everywhere on \(X\).  When \(V\) is continuous, this equality holds
everywhere.
\end{lemma}

\begin{proof}
Write \(\omega=\tau_1+\mathrm{i}\tau_2\), where
\(\tau_1=-\tfrac12\) and \(\tau_2=\tfrac{\sqrt3}{2}\), and define
\begin{equation}\label{eq:theta-density}
 \Xi(z):=\mathrm{e}^{-2\pi(\operatorname{Im}z)^2/\tau_2}
 |\vartheta(z)|^2.
\end{equation}
This is periodic by the transformation laws of \(\vartheta\).  For
\(w=z(k)\), the relation \(k=2\pi\mathrm{i}w/\tau_2\) and
\eqref{eq:Fk} give
\begin{equation}\label{eq:Fk-density}
 |F_k(z)|^2=
 \mathrm{e}^{2\pi(\operatorname{Im}w)^2/\tau_2}
 \frac{\Xi(z-w)}{\Xi(z)}.
\end{equation}

The function \(\Xi\) has a zero of order two at \([0]\) and no other zero
on \(X\).  Lemma~\ref{lem:chiral-zero-mode} gives
\(p(z),p(-z)=O(|z|)\) there.  It follows that
\[
 f_V(z):=
 \frac{\operatorname{Im}\!\left(V(z)p(-z)\overline{p(z)}\right)}{\Xi(z)}
\]
belongs to \(L^2(X)\).  If \(a_V(k)=0\) for all \(k\), equations
\eqref{eq:aformula} and \eqref{eq:Fk-density} imply
\begin{equation}\label{eq:theta-convolution}
 \int_X\Xi(z-w)f_V(z)\,\mathrm{d}m(z)=0
 \quad(w\in X),
\end{equation}
because \(k\mapsto z(k)\) maps the Bloch torus onto \(X\).

None of the Fourier coefficients of \(\Xi\) vanishes.  To see this, write
\(z=x+y\omega\) and use the characters
\(\mathrm{e}^{2\pi\mathrm{i}(rx+sy)}\), \((r,s)\in\Z^2\).  Substitution of
the theta series
\[
 \vartheta_1(z\mid\omega)
 =-\mathrm{i}\sum_{n\in\Z}(-1)^n
 \mathrm{e}^{\pi\mathrm{i}\omega(n+1/2)^2
             +2\pi\mathrm{i}(n+1/2)z}
\]
into \eqref{eq:theta-density} gives
\begin{equation}\label{eq:theta-Fourier}
 \widehat\Xi(r,s)=c_\vartheta\,\sigma_{r,s}
 \exp\!\left(-\frac{\pi\tau_2r^2}{2}
 -\frac{\pi(s-\tau_1r)^2}{2\tau_2}\right),
 \quad \sigma_{r,s}\in\{-1,1\},
\end{equation}
where \(c_\vartheta>0\) depends only on the normalization of area and of the
theta function.  Indeed, the integration in \(x\) fixes the difference of
the two summation indices to be \(r\).  The remaining sum and the integration
in \(y\) unfold to
\[
 \int_\R \mathrm{e}^{-2\pi\tau_2u^2}
 \mathrm{e}^{2\pi\mathrm{i}(r\tau_1-s)u}\,\mathrm{d}u,
\]
which proves \eqref{eq:theta-Fourier}.  Taking Fourier coefficients in
\eqref{eq:theta-convolution} now gives \(\widehat f_V(r,s)=0\) for every
\((r,s)\).  Thus \(f_V=0\), and hence
\(\operatorname{Im}\!\left(V(z)p(-z)\overline{p(z)}\right)=0\), almost
everywhere.  The last statement follows from continuity.
\end{proof}

Combining \eqref{eq:scalarA}, \eqref{eq:aformula}, and
Lemma~\ref{lem:theta-averages}, we obtain
\begin{equation}\label{eq:vanishing-compression}
 \begin{aligned}
 A_V(k)=0\ \text{for every }k
 &\quad\Longleftrightarrow\quad
 a_V(k)=0\ \text{for every }k\\
 &\quad\Longleftrightarrow\quad
 \operatorname{Im}\!\left(V(z)p(-z)\overline{p(z)}\right)=0
 \quad\text{a.e. on }X.
 \end{aligned}
\end{equation}
Indeed, the first equivalence follows from \(A_V(k)=a_V(k)I_2\), while the
second follows from \eqref{eq:aformula} and
Lemma~\ref{lem:theta-averages}.  Consequently, the invertibility criterion
holds precisely when
\(\operatorname{Im}(V(z)p(-z)\overline{p(z)})\) is nonzero on a set of
positive measure in \(X\).

This equivalence allows us to characterize exactly what
potentials exhibit exactly flat bands for every Bloch momentum and every real nonchiral coupling.

\begin{theorem}[Characterization of exact flatness]
\label{thm:characterization}
Let \(U\) be a nonzero real-analytic tunnelling potential satisfying
\eqref{eq:Usym}, let \(\alpha_*\) be a simple real magic parameter for
\(U\), and let \(V\) be an admissible nonchiral tunnelling potential.  Then
\[
 \dim_\C\ker H_k(\alpha_*,\lambda)\geq2
 \quad(k\in\C/\Lambda^*,\ \lambda\in\R)
\]
if and only if \(V=V[f]=fV_0\) for some
\(f\in\mathcal F_{\rm sym}\).  Equivalently, these conditions hold if and
only if
\[
 A_V(k)=0\quad(k\in\C/\Lambda^*).
\]
\end{theorem}

\begin{proof}
By \eqref{eq:vanishing-compression}, \(A_V(k)=0\) for every \(k\) if and
only if \(\operatorname{Im}(V(z)p(-z)\overline{p(z)})=0\) almost everywhere
on \(X\).  Since \(V_0\) is nontrivial and real analytic, its zero set has
measure zero.  Away from this set,
\[
 f(z):=\frac{V(z)}{V_0(z)}
 =\frac{V(z)p(-z)\overline{p(z)}}{|p(z)p(-z)|^2}
\]
is real.  The symmetries of \(V\) and \(V_0\), together with
\(fV_0=V\in L^\infty(X)\), give \(f\in\mathcal F_{\rm sym}\) after extending
\(f\) across the zero set.  Thus \(A_V(k)=0\) for every \(k\) implies
\(V=V[f]\).  Conversely, for \(V=V[f]\),
\(V(z)p(-z)\overline{p(z)}=f(z)|p(z)p(-z)|^2\) is real, so
\eqref{eq:vanishing-compression} gives \(A_V(k)=0\) for every \(k\).

Theorem~\ref{thm:exact} gives the kernel inequality for \(V=V[f]\), with the
case \(V[0]\equiv0\) given by Lemma~\ref{lem:chiral-zero-mode}.  Conversely, suppose
the kernel inequality holds.  If \(A_V(k_0)\neq0\) for some \(k_0\), then
\(A_V(k_0)\) is invertible because it is a scalar multiple of \(I_2\).
Proposition~\ref{prop:compression} would then imply that
\(H_{k_0}(\alpha_*,\lambda)\) has no zero eigenvalue for all sufficiently
small nonzero \(\lambda\), a contradiction.  Hence \(A_V(k)=0\) for every
\(k\), completing the proof.
\end{proof}

We now specialize to \(U=U_{\rm BM}\) and \(V=V_{\rm BM}\).  By
\eqref{eq:vanishing-compression}, it suffices to show that
\[
 \operatorname{Im}\!\left(V_{\rm BM}(z)p(-z)\overline{p(z)}\right)
 \not\equiv0.
\]
This function is real analytic and hence is nonzero on an open set once it
is not identically zero; we prove the latter by comparing its Taylor
coefficients at \(z=0\).

\begin{proof}[Proof of Proposition~\ref{prop:standard-BM}]
Suppose, to the contrary, that
\begin{equation}\label{eq:gBM-zero}
 \operatorname{Im}\!\left(
 V_{\rm BM}(z)p(-z)\overline{p(z)}\right)=0
 \quad(z\in X).
\end{equation}

By Lemma~\ref{lem:chiral-zero-mode}, specifically
\eqref{eq:u0-local-factor}, near \(0\) we have
\(\vec u_0(z)=z\vec w(z)\) with \(\vec w(0)\neq0\).  Thus \(p\) vanishes to
first order at \(0\), and the rotation law in \eqref{eq:p-sym} forces its
linear term to be \(cz\), with \(c\neq0\).  Multiplying
the zero mode by a constant does not affect whether \eqref{eq:gBM-zero}
holds, so we normalize this coefficient to be one.  Equations
\eqref{eq:u0} and \eqref{eq:u0-reflection} give
\(\overline{p(-\bar z)}=\rho p(z)\).  After normalization,
\(p(z)=z+O(|z|^2)\), and hence
\[
 \overline{p(-\bar z)}=-z+O(|z|^2),
 \quad
 \rho p(z)=\rho z+O(|z|^2).
\]
Thus \(\rho=-1\).  Write the real-analytic expansion as
\[
 p(z)=\sum_{m,n\geq0}c_{mn}z^m\bar z^n.
\]
The identity \(\overline{p(-\bar z)}=-p(z)\) gives
\[
 \overline{c_{mn}}=(-1)^{m+n+1}c_{mn}.
\]
Thus coefficients of odd total degree \(m+n\) are real, while those of even
total degree are purely imaginary.  In particular, the coefficient of
\(z^4\) can be written as \(\mathrm{i}\beta\), where \(\beta\in\R\).

Put \(K=4\pi/3\) and \(t=\eps\alpha_*/2\).
The finite Fourier sums in \eqref{eq:BMpotentials} give
\begin{align*}
 U_{\rm BM}(z)={}&\frac{3K^2}{2}z
 +\frac{3\mathrm{i}K^3}{8}\bar z^2
 -\frac{3K^4}{16}z^2\bar z
 -\frac{\mathrm{i}K^5}{32}z\bar z^3
 -\frac{\mathrm{i}K^5}{128}z^4+O(|z|^5),\\
 V_{\rm BM}(z)={}&3-\frac{3K^2}{4}z\bar z
 -\frac{\mathrm{i}K^3}{16}(z^3+\bar z^3)
 +\frac{3K^4}{64}z^2\bar z^2\\
 &\hspace{25mm}
 +\frac{\mathrm{i}K^5}{256}(z^4\bar z+z\bar z^4)+O(|z|^6).
\end{align*}
The first row of \(D(\alpha_*)\vec u_0=0\) is
\begin{equation}\label{eq:p-local-equation}
 \partial_{\bar z}p(z)=tU_{\rm BM}(z)p(-z).
\end{equation}
Comparing coefficients in this equation determines every Taylor coefficient
with a positive power of \(\bar z\) from the preceding ones.  Through the
order needed below, this gives
\begin{align}
 p(z)={}&z-\frac{3K^2t}{2}z^2\bar z+\mathrm{i}\beta z^4
 -\frac{\mathrm{i}K^3t}{8}z\bar z^3
 +\frac{3K^4t(12t+1)}{32}z^3\bar z^2 \notag\\
 &+\frac{\mathrm{i}K^2t(K^3+192\beta)}{128}z^5\bar z
 +\frac{\mathrm{i}K^5t(12t+1)}{128}z^2\bar z^4
 +O(|z|^7).\label{eq:p-BM-expansion}
\end{align}

Substitution into \eqref{eq:gBM-zero} gives
\begin{align}
 &\operatorname{Im}\!\left(
 V_{\rm BM}(z)p(-z)\overline{p(z)}\right) \notag\\
 &\quad=-\frac{6K^3t-K^3-48\beta}{16}
 (z^4\bar z+z\bar z^4)\notag\\
 &\quad\quad
 +\frac{K^2(216K^3t^2-12K^3t-K^3-192\beta)}{256}
 (z^5\bar z^2+z^2\bar z^5)+O(|z|^8).
 \label{eq:gBM-expansion}
\end{align}
The degree-five coefficient would have to vanish, giving
\(\beta=K^3(6t-1)/48\).  With this value of \(\beta\), the degree-seven
coefficient becomes
\[
 \frac{3K^5}{256}(72t^2-12t+1)
 =\frac{3K^5}{256}
 \left(72\left(t-\frac1{12}\right)^2+\frac12\right)>0,
\]
since \(t\) is real.  This contradicts \eqref{eq:gBM-zero}.  Therefore
\(a_{V_{\rm BM}}\) is not identically zero.
\end{proof}

Together with Proposition~\ref{prop:compression} and
Theorem~\ref{thm:discrete-couplings}, this proves the standard-BM
conclusions stated in Sec.~\ref{sec:model}.

\section{Discussion and future directions}\label{sec:discussion}

At a fixed simple magic parameter \(\alpha_*\), we have constructed a family
of admissible nonchiral potentials \(V[f]=fV_0\) for which at least two
zero-energy states remain at every Bloch momentum for every real nonchiral
coupling \(\lambda\).  The line-bundle and index argument proves this exact
flatness without a small-coupling assumption, while
Theorem~\ref{thm:characterization} shows that every admissible potential
with this property is of this form.   In the opposite direction, the
invertibility criterion obstructs exact flatness, and for the standard
first-harmonic BM potential we established that exact flatness can not occur except for at most a discrete set of nonchiral couplings. 

We have characterized nonchiral tunneling potentials for which the flat bands survive for every \(\lambda\)  at the magic parameter \(\alpha_*\).  It does not classify potentials that produce exact \textcolor{black}{flatness} at one isolated nonzero coupling, nor does it address exact flatness at a value of \(\alpha\) different from \(\alpha_*\).  A complete  characterization of exactly flat bands  in the two-parameter plane \((\alpha,\lambda)\) remains to be understood.  For the standard BM model, it is natural to ask whether the discrete exceptional set in Theorem~\ref{thm:discrete-couplings} can be removed, so that exact flatness is excluded for every \(\lambda\neq0\).  It would also be
useful to extend the criterion and the standard-BM conclusion to
higher-order continuum models \cite{QuinnKongLuskinWatson}, relaxed
tunnelling potentials \cite{NamKoshino}, and models incorporating strain.

The potentials \(V[f]\) that lead to nonchiral exact flatness are defined at the continuum level.  A microscopic realization would require an atomistic interlayer hopping model, together with the stacking geometry and lattice relaxation, to produce a same-sublattice tunnelling potential satisfying
\(V(z)=f(z)p(z)\overline{p(-z)}\).  It is therefore interesting to understand whether this could ever possibly arise from realistic microscopic configurations, or whether the exact flatness is a purely mathematical construction.

\appendix
\section{Numerical illustration}\label{app:numerics}

We numerically test Theorem~\ref{thm:exact} for \(V=V_0\), using different
plane-wave cutoffs to construct the potential and evaluate the Hamiltonian.

\subsection{Numerical setup}

At cutoff \(N\), the reciprocal-lattice momenta are
\[
 G=mB_1+nB_2,\quad |m|,|n|\leq N.
\]
We use two cutoffs.  At the \emph{construction cutoff} \(N_{\rm s}\), we
compute a normalized chiral zero mode at \(k=0\) and obtain the Fourier
coefficients of \(V_0(z)=p(z)\overline{p(-z)}\) from its first component by
discrete convolution, normalized by \(\sum_G|\widehat V_0(G)|^2=1\).  Keeping
this \(V_0\) fixed, we diagonalize a separately assembled Hamiltonian at the
\emph{evaluation cutoff} \(N_{\rm t}=N_{\rm s}+1\); \(V_0\) is not recomputed
at \(N_{\rm t}\).  Thus flatness is tested in a larger basis than the one used
to construct the potential, checking against a same-truncation identity.

For the calculation shown below, \(N_{\rm s}=6\), \(N_{\rm t}=7\),
\(\alpha_*\approx0.585663558390\), and \(\lambda=1\).  The smallest singular
value of the source chiral matrix is \(2.07\times10^{-14}\).

\subsection{Results}

We sample the momentum parallelogram \(k=sB_1+tB_2\),
\(-\tfrac12\leq s,t\leq\tfrac12\), on a \(13\times13\) grid and also sample
364 points along \(\Gamma\)--\(K\)--\(M\)--\(\Gamma\), where
\(\Gamma=0\), \(K=(2B_1-B_2)/3\), and \(M=B_1/2\).  If \(E_{+1}(k)\) and
\(E_{-1}(k)\) are the two central eigenvalues, set
\[
 r(k):=\max\{|E_{+1}(k)|,|E_{-1}(k)|\}.
\]
Across all sampled momenta, \(r(k)\leq2.61\times10^{-11}\), while the nearest
neighboring band remains at least \(2.08\) away from zero.
Figure~\ref{fig:band-surface} displays \(r(k)\) over the Bloch-torus grid and
the ten bands along the high-symmetry path.

The convergence in Table~\ref{tab:numerical-convergence} under successive
source and target cutoffs checks against a same-truncation algebraic artifact.

\begin{table}[H]
\centering
\begin{tabular}{c c c c}
\hline
\(N_{\rm s}\) & \(N_{\rm t}\) & target dimension & \(\max_k r(k)\)\\
\hline
3 & 4 & 324 & \(1.97\times10^{-5}\)\\
4 & 5 & 484 & \(5.87\times10^{-8}\)\\
5 & 6 & 676 & \(1.23\times10^{-9}\)\\
6 & 7 & 900 & \(2.61\times10^{-11}\)\\
\hline
\end{tabular}
\caption{Cross-cutoff convergence of the maximum central-band residual over
the Bloch-torus grid and high-symmetry path.}
\label{tab:numerical-convergence}
\end{table}

\begin{figure}[H]
\centering
\includegraphics[width=\linewidth]{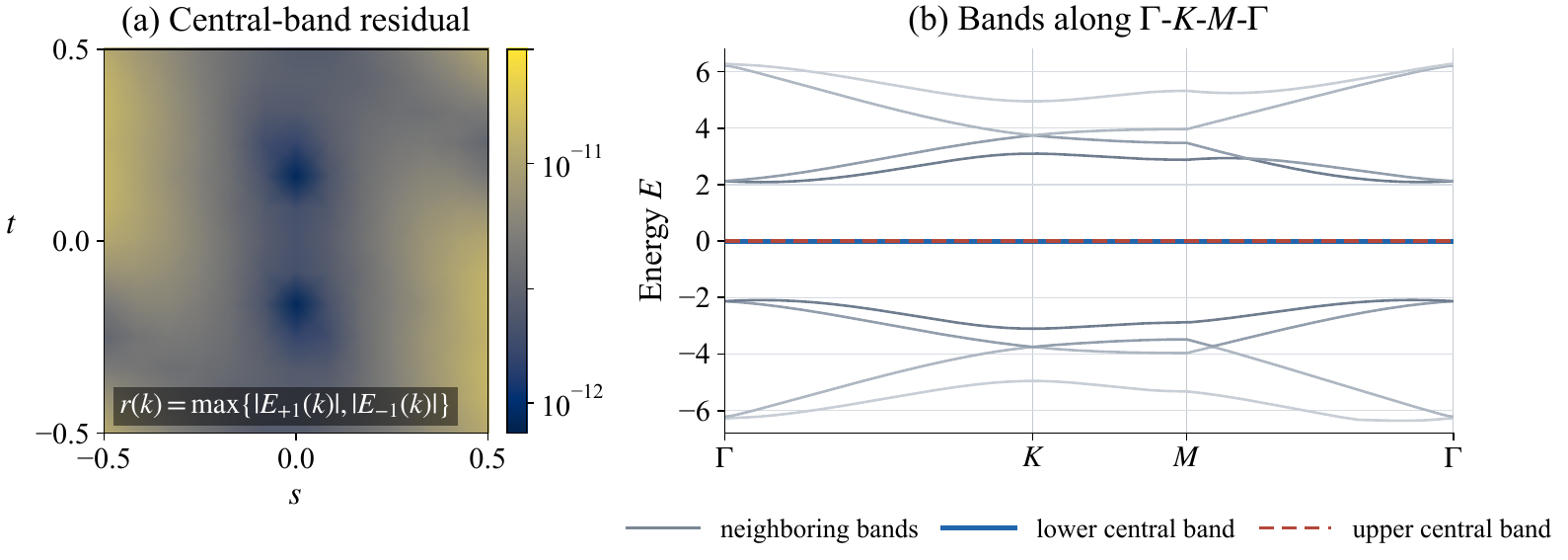}
\caption{Numerical check for \(V_0\).  Panel (a) shows the central-band
residual \(r(k)=\max\{|E_{+1}(k)|,|E_{-1}(k)|\}\) on the \(13\times13\)
Bloch-torus grid.  Panel (b) shows the
ten central bands along \(\Gamma\)--\(K\)--\(M\)--\(\Gamma\); the solid blue
and dashed red curves are the two central bands and coincide at the scale of
the plot.  The potential is constructed at cutoff \(N_{\rm s}=6\) and
evaluated at cutoff \(N_{\rm t}=7\).}
\label{fig:band-surface}
\end{figure}

\section*{Data availability}

The numerical data supporting Figures~\ref{fig:v0} and
\ref{fig:band-surface} and Table~\ref{tab:numerical-convergence} are
reproducible using the code \texttt{nonchiral\_flat\_band.py} accompanying
the paper.

\section*{AI usage statement}

{ ChatGPT and Codex were used to assist with suggesting proof strategies, preparing the \LaTeX{} manuscript, and verifying correctness.}  The authors take full responsibility for the content of the manuscript.

We also used Codex to assist in searching for counterexamples to Open Problem 3 of \cite{ZworskiSurvey}, motivated by recent successes in using AI tools to find counterexamples to mathematical conjectures. In our case, however, Codex repeatedly questioned whether such counterexamples should exist and offered various plausible physical heuristics against them throughout the entire process. 

At that point, most of the argument in Section~\ref{sec:construction} was already in place: we had identified the invertibility criterion as a natural obstruction to flatness, but saw no reason to expect it to hold for every admissible potential. { The authors therefore judged that a counterexample remained very plausible and were not moved by Codex's judgment.}  This experience serves as a reminder that AI-assisted mathematical exploration still depends crucially on human judgment and the correct intuition, both mathematical and physical, without which we cannot expect to make progress.  
\bibliographystyle{unsrt}
\IfFileExists{exact_flat_bands_beyond_chiral_limit.bib}
  {\bibliography{exact_flat_bands_beyond_chiral_limit}}
  {\bibliography{outputs/pdf/exact_flat_bands_beyond_chiral_limit}}

\end{document}